%% file: root.tex
\documentclass[letterpaper, 10 pt, conference]{ieeeconf}  

\IEEEoverridecommandlockouts                              
\input{utils/preamble}

\usepackage{graphics} 
\usepackage{color}
\newcommand{\revisedText}[2]{#2}

\title{\LARGE \bf Stochastic Nonlinear Model
Predictive Control with\\ Gaussian Mixture Uncertainty Propagation
}
\author{Konstantinos Prattis, Luca Laurenti$^{1}$, and Azita Dabiri
\thanks{
All authors are with the Delft Center for Systems and Control, Delft University of Technology, Delft, The Netherlands \texttt{\{k.prattis, l.laurenti, a.dabiri\}@tudelft.nl,}$^{1}$is also with AI4I, Turin, Italy.}
}

\begin{document}
\setlength{\abovedisplayskip}{4pt plus 2pt minus 1pt}
\setlength{\belowdisplayskip}{4pt plus 2pt minus 1pt}
\setlength{\abovedisplayshortskip}{2pt plus 1pt}
\setlength{\belowdisplayshortskip}{2pt plus 1pt}

\bstctlcite{IEEEexample:BSTcontrol}
\maketitle
\thispagestyle{empty}
\pagestyle{empty}

\newpage
\begin{abstract}
We propose a novel Stochastic Nonlinear Model Predictive Control (SNMPC) framework for nonlinear systems with additive noise. Building on recent advances in nonlinear uncertainty propagation, we show that the state distribution of the system can be tractably approximated over time by Gaussian mixture distributions, with formal error bounds in Wasserstein distance. This representation yields closed-form expressions for expected costs and chance constraints, which become exact for affine constraints and exact up to a constant for quadratic costs. Consequently, the resulting control problem can be solved efficiently via nonlinear programming, while providing formal open-loop guarantees of correctness and asymptotic optimality. Experiments on a set of benchmarks demonstrate that the proposed approach compares favorably with existing methods in nonlinear settings with multi-modal disturbances, where standard approaches lead to poorly scaled solutions and unsafe or overly conservative control actions.\looseness=-1
\end{abstract}

\section{INTRODUCTION}
\input{sections/introduction}
\section{PRELIMINARIES}
\vspace{-1.5mm}
\input{sections/preliminaries/notation}
\input{sections/preliminaries/wasser}
\section{PROBLEM FORMULATION}
\input{sections/problem_form}
\section{FORMAL UNCERTAINTY PROPAGATION WITH GAUSSIAN MIXTURES} \label{sec_mix_prop_all}
\input{sections/gmm_unc_prop}
\section{TRACTABLE SNMPC REFORMULATION} \label{sec_reform_all}
\input{sections/mixmpc/cost_and_cdf_derv}
\subsection{Resulting optimization problem}
\input{sections/mixmpc/formulation}
\section{IMPLEMENTATION \& NUMERICAL RESULTS} \label{sec_imlp_all}
\input{sections/results/implementation}

\subsection{Simulations and numerical examples}
\input{sections/results/simulations}

\section{CONCLUSION AND FUTURE WORK} \label{sec_conc}
\input{sections/conclusion}

\bibliographystyle{IEEEtran}
\bibliography{IEEEabrv, references}

\end{document}

%% file: utils/preamble.tex
\usepackage{graphicx} 
\usepackage{amsmath} 
\usepackage{amssymb}  

\usepackage[long]{optidef} 

\usepackage{booktabs}

\usepackage[hidelinks]{hyperref}
\usepackage{url}
\usepackage{cite}

\usepackage{utils/commands}

%% file: sections/introduction.tex
Modern autonomous systems are inherently stochastic \cite{unifyingstochastic}. On one hand, the rapid evolution of robots, drones, and autonomous vehicles has accelerated the transition from controlled factory environments to real-world settings characterized by significant uncertainty. On the other hand, the increasing use of data-driven methods and machine learning to learn models and policies enables the practical solution of complex tasks, but introduces additional uncertainty due to the reliance on statistical approximations and finite data \cite{Hewing2025Learning-BasedControl}. Finally, noise in sensor measurements and process disturbances necessitates the use of state estimation techniques, which further contribute to the stochastic nature of the system \cite{doucetfg01}.
Consequently, achieving optimal performance and provably safe operation in real-world settings often requires explicit modeling of uncertainty and control algorithms capable of accounting for stochastic effects.\looseness=-1

Model Predictive Control (MPC) is one of the most widely used control frameworks in modern control engineering \cite{book}. In the presence of bounded uncertainty, Robust MPC is a variant of standard MPC, which can provide safety guarantees, but often at the cost of conservatism \cite{9044326}. A less conservative alternative, applicable also to unbounded disturbances, is Stochastic Model Predictive Control (SMPC) \cite{Mesbah2016StochasticResearch}. As an extension of deterministic MPC, SMPC operates in a receding horizon fashion using a stochastic system model, but optimizes the expected cost while enforcing probabilistic constraints \cite{Farina2016StochasticReview}.
While SMPC for linear systems with Gaussian noise is a well-understood problem \cite{Farina2016StochasticReview}, its extension to nonlinear systems and non-Gaussian uncertainties remains a challenging open problem, which is the object of this paper. 

Existing approaches for nonlinear SMPC typically rely on approximations that lack explicit error bounds and therefore lack guarantees of performance or correctness \cite{Mesbah2016StochasticResearch}. This difficulty arises primarily from the intractability of propagating uncertainty through nonlinear dynamics \cite{LANDGRAF2023100905}. For example, methods such as those based on moments propagation (e.g. \cite{Kabzan2019Learning-basedRacing}), approximate the state distribution by propagating only its first few moments via local linearization or the Unscented Transform (UT) \cite{1997SPIE.3068..182J}. Even though success of moment-matching has been empirically verified \cite{Kabzan2019Learning-basedRacing}, such heuristics neglect higher-order effects and additional approximation errors are inevitably introduced.
%
In contrast, more mathematically rigorous methods come with the cost of high computational complexity. Methods based on Polynomial Chaos Expansion (PCE) (e.g. \cite{Fagiano2012NonlinearExpansions}) prove convergence in the $L_2$ sense, but suffer from Gibbs non-smoothness phenomena and the curse of dimensionality \cite{2002SJSC...24..619X}. Similarly, scenario-based approaches (e.g. \cite{SCHILDBACH20143009}) are naturally accompanied by distribution-free convergence guarantees, but require a high sample complexity to safeguard from tail events and reach suitable confidence levels for safety-critical systems \cite{1632303}.

To address these challenges, this paper leverages recent advances in nonlinear uncertainty propagation \cite{Adams2025FormalNoise} to derive a tractable SMPC reformulation with formal guarantees. We build on the framework proposed in \cite{Adams2025FormalNoise} to approximate the state distribution of the system over time using Gaussian mixtures, equipped with explicit error bounds in Wasserstein distance \cite{CedricVillaniOptimalNew}. \revisedText{blue}{Although the Wasserstein distance is central to our approach, our setting differs from standard distributionally robust MPC (e.g.  \cite{mcallisterDistributionallyRobustModel2025}), where an ambiguity set models genuine uncertainty about the disturbance distribution, while here the Wasserstein radius mirrors a controlled approximation error. Moreover, while mixture and ensemble-based methods, e.g. Gaussian-sum filtering \cite{ALSPACH1974271} and particle filtering \cite{doucetTutorialParticleFiltering}, and sigma-point methods \cite{1997SPIE.3068..182J}, have been used for prediction, our construction differs as we also provide a  formal quantification of the resulting distributional error.}

The derived Gaussian mixture representations enable the derivation of tractable locally-analytic expressions for both chance constraints and expected cost, obtained via local linearization around each mixture component. Using tools from optimal transport theory \cite{CedricVillaniOptimalNew}, we establish formal guarantees for the proposed reformulation and quantify the conservativeness introduced by the approximations. Furthermore, we show that, as the number of mixture components increases, the solution of the resulting optimization problem converges to that of the original problem.
The resulting optimization problem is solved efficiently using a Sequential Quadratic Programming (SQP) scheme. Empirical results on two nonlinear SMPC benchmarks demonstrate the effectiveness of the proposed approach under both unimodal and bimodal disturbance distributions.
In summary, the main contributions of this work are:\looseness=-1
\begin{itemize}
\item a tractable SMPC reformulation based on Gaussian mixture uncertainty propagation, amenable to efficient solution via nonlinear programming,
\item formal guarantees on approximation error and constraint satisfaction \revisedText{blue}{for the case of affine constraints, quadratic costs and Lipschitz dynamics}, including bounds on the conservativeness of the reformulation and convergence guarantees to the original stochastic problem,
\item empirical validation on nonlinear benchmarks demonstrating competitive performance.
\end{itemize}

%% file: sections/preliminaries/notation.tex
We use $\mathcal{W}^m \subset \R^m$ to denote the $(m-1)$-dimensional simplex $\{w_1, \ldots, w_m \in \R_+ : \sum_{i=1}^{m} w_i = 1\}$.
The set of square symmetric positive definite matrices of size $n$ is denoted by $\mathbb{S}_{++}^n$. For a vector $x \in \R^n$ and a  matrix $A\in \mathbb{S}_{++}^{n}$, we use $\|x\|_A$ for $\sqrt{x^T A x}$.
For a set $\mathcal{X}$, the set $\mathcal{X}^n$ denotes the $n$-times Cartesian product of the set with itself.
For $\mathcal{X} \subseteq \mathbb{R}^n$ we use $\mathcal{P}(\mathcal{X})$ for the set of Borel
probability measures on $\mathcal{X}$, and $\mathcal{P}_\rho(\mathcal{X})$ for those with finite
$\rho$-moments.
We use $\delta_x$ for the Dirac distribution centered at $x \in \mathcal{X}$.
For $Q \in \mathcal{P}(\mathcal{X})$, we use $\mathbb{E}_Q[\cdot]$ for the expectation and $\mathbb{Q}(\cdot)$ for the probability operators. The convolution operator is denoted with '$*$' and the push-forward operator with '$\#$'.
For $\mu \in (\mathbb{R}^{n})^{m}$, $\Sigma \in (\mathbb{S}^{n}_{++})^{m}$ and
$w \in \mathcal{W}_m$, the Gaussian mixture $GM(w,\mu,\Sigma) := \sum_{i=1}^{m} w_i\,
\mathcal{N}(\mu_i,\Sigma_i)$ is the convex combination of the Gaussians
$\mathcal{N}(\mu_i,\Sigma_i)$ with mean $\mu_i$ and covariance $\Sigma_i$; its number of components
is denoted by $|GM(w,\mu,\Sigma)|$.
For the Gaussian mixture with $\Sigma_i= \Sigma_0$, for all $i=1, \ldots, m$, and $\Sigma_0\in \mathbb{S}_{++}^n$, we use the notation $\mathrm{GM}(w, \mu, \Sigma_0)$. The set of all such mixtures in $\R^n$, for a given $\Sigma_0 \in \mathbb{S}_{++}^n$ and $\left|\mathrm{GM}(w, \mu, \Sigma_0)\right|=m$, will be denoted by $\mathcal{G}_m^{\Sigma_0}(\R^n)$.

%% file: sections/preliminaries/wasser.tex
Finally, let $\mathcal{X} \subseteq \R^n$ and $P, Q \in \mathcal{P}_\rho(\mathcal{X})$ be any two probability measures. For any integer $\rho \ge 1$, their \textit{$\rho$-Wasserstein distance} is defined as \cite{CedricVillaniOptimalNew}:
\begin{equation}
    W_\rho(P, Q) = \left(\inf_{\gamma \in \Gamma(P, Q)} \mathbb{E}_{(x, y)\sim\gamma}\big[\|x - y\|^\rho\big] \right)^{\frac{1}{\rho}}
\end{equation}
where $\Gamma(P, Q) \subset \mathcal{P}(\mathcal{X}\times\mathcal{X})$ is the set of all joint distributions with marginals $P$ and $Q$.

%% file: sections/problem_form.tex
We consider discrete-time stochastic systems governed by the following stochastic difference equation:
\begin{equation}
\label{eq_dynamics}
x_{t+1} = f\left(x_t, u_t\right) + \eta_t,
\end{equation}
where $t\in \mathbb{N}_+$ is the discrete time, $x_t \in \mathbb{R}^{n_x}$ is the state of the system at time $t$, $u_t\in \mathbb{R}^{n_u}$ is the control input, and the possibly nonlinear function $f:\R^{n_x}\times\R^{n_u} \rightarrow \R^{n_x}$ in \eqref{eq_dynamics} is assumed to be piecewise differentiable. Full-state measurements are available at every step, and the additive disturbance
$\eta_t \in \mathbb{R}^{n_x}$ is i.i.d. according to a known $P_\eta \in \mathcal{P}(\mathbb{R}^{n_x})$.\looseness=-1 
\begin{assumption} \label{asu_noise_wasser}
The disturbance distribution $P_\eta$ is a Gaussian mixture with equal and diagonal covariance $\Sigma_\eta$ across components, i.e. $P_\eta \coloneqq \mathrm{GM}(w_\eta, \mu_\eta, \Sigma_\eta) \in \mathcal{G}^{\Sigma_\eta}_{m_\eta}(\R^{n_x})$.
\end{assumption}
\begin{remark}
We focus for simplicity on Gaussian mixtures with equal covariances. \revisedText{blue}{However, note that any distribution with finite second moments can be approximated arbitrarily well in $W_2$ by this class \cite{CedricVillaniOptimalNew}, provided both the number of components can grow and the common covariance can vanish. Our framework can easily be extended to account for the resulting Wasserstein error bounds via the techniques in \cite{Adams2025FormalNoise}. In practice however, the curse of dimensionality applies: one would need $\mathcal{O}(\epsilon^{-n_x})$ components to reach a $W_2$ error of order $\epsilon$ \cite{Graf2000FoundationsDistributions}.\looseness=-1}
\end{remark}

We aim to control \eqref{eq_dynamics} using the policy $u_t = \pi(x_t, \theta_t)$, where $\pi:\R^{n_x}\times \Theta \rightarrow \R^{n_u}$ is a fixed differentiable policy
\footnote{
This covers open-loop control, $\pi(x,\theta)=\theta$, and state feedback,
$\pi(x,\theta)=Kx+v$ with $\theta=[\mathrm{vec}(K);v]$ or $\theta=v$~\cite{Kohler2023OnFramework}.
}
parametrized by $\theta_t \in \Theta \subseteq \R^{n_\theta}$. In particular, the goal consists of synthesizing an optimal strategy for selecting the parameter vector $\theta_t$ in a receding-horizon setting, by minimizing a given cost while enforcing
probabilistic state and input constraints over a given finite prediction horizon. Before formally introducing  the problem, we need to first introduce some necessary notation. In particular, we denote by $x_{k|t}$ the predicted state $k$ steps after the observed state $x_t$. Then, $P_{x_{k|t}}\in \mathcal{P}(\R^{n_x})$, the distribution of $x_{k|t}$ given $x_{0|t} = x_t$, is induced by $(\eta_{0|t}, \ldots,\eta_{k-1|t}) \sim P_\eta^k$ and $(\theta_{0|t}, \ldots, \theta_{k-1|t}) \in \Theta^k$ via the recursion $P_{x_{k+1|t}} = \mathbb{E}_{P_{x_{k|t}}}[T(x_{k|t}, \theta_{k|t})]$, where $T:\R^{n_x}\times\Theta\rightarrow \mathcal{G}^{\Sigma_\eta}_{m_\eta}(\R^{n_x})$
is the one-step transition kernel given by \(T(x, \theta) \coloneqq \delta_{f_\pi(x, \theta)} * P_\eta \), with $f_\pi(x, \theta) \coloneqq f(x, \pi(x, \theta))$. Note that, because of the additivity of the noise in \eqref{eq_dynamics}, $T$ is a Gaussian mixture distribution. 
We can now define differentiable functions $h:\R^{n_x}\times\R^{n_u}\rightarrow\R$, $h_N:\R^{n_x}\rightarrow\R$ to be the nominal constraints\footnote{Note that Problem \ref{problem_snmpc} and our proposed method can naturally be extended to include multiple constraints $h_{\pi,j}$.}, and the twice differentiable functions $\ell:\R^{n_x}\times\R^{n_u}\rightarrow\R_+$, $\ell_N:\R^{n_x}\rightarrow\R_+$ to represent the nominal costs. Denoting by $h_\pi(x_{k|t}, \theta_{k|t})\coloneqq h(x_{k|t}, \pi(x_{k|t}, \theta_{k|t}))$ and $\ell_\pi(x_{k|t}, \theta_{k|t}) \coloneqq \ell(x_{k|t}, \pi(x_{k|t}, \theta_{k|t}))$ the respective functions evaluated under $\pi$, we can finally state our problem:  
\begin{problem}[Stochastic Nonlinear MPC (SNMPC)]
\label{problem_snmpc}
Given a prediction horizon $N\in\mathbb{N}_+$, satisfaction probabilities $1-\epsilon_k \in (0,1)$, $k=1,\ldots, N$, and the cost $J_{P|t}(\theta_{\cdot|t}):=\mathbb{E}_{P_{\eta}^N}\big[\sum_{k=1}^{N-1} \ell_\pi(x_{k|t}, \theta_{k|t}) + \ell_N(x_{N|t})\big]$, 
find $\theta^*_{\cdot|t} = (\theta_{0|t}^*, \ldots, \theta^*_{N-1|t}) \in \Theta^N$, such that:
\begin{argmini!}|s|[2]<b>
{\substack{\theta_{\cdot|t} \in \Theta^{N} 
}}
{ J_{P|t}(\theta_{\cdot|t}) \label{eq_smpc_obj}}
{\label{eq_smpc}}
{\theta^*_{\cdot|t}\in}
\addConstraint{x_{0|t}}{= x_t \label{eq_smpc_constr_x0}}
\addConstraint{x_{k+1|t}}{= f_\pi(x_{k|t}, \theta_{k|t}) + \,\eta_{k|t} \label{eq_smpc_constr_dyn}}
\addConstraint{\mathbb{P}_{x_{k|t}}\big(h_\pi(x_{k|t}, \theta_{k|t}) \le 0\big)}{\ge 1 - \epsilon_k \label{eq_smpc_constr_prob_k}}
\addConstraint{\mathbb{P}_{x_{N|t}}\big(h_N(x_{N|t})\le 0\big) \ge 1 - \epsilon_N .\label{eq_smpc_constr_prob_f}}
\end{argmini!}
\end{problem}
\vspace{2mm}

Unfortunately, despite its importance, solving the SNMPC Problem \ref{problem_snmpc} is generally not possible \cite{book}. This is due to the nonlinearities in $f$ and the non-Gaussian nature of the noise, which do not allow for closed form expressions for the distribution of $x_{k|t}$ for $t>1$ \cite{LANDGRAF2023100905}. Consequently, obtaining closed-form expressions for both $J_{P|t}(\theta_{\cdot|t})$ and chance constraints \eqref{eq_smpc_constr_prob_k}--\eqref{eq_smpc_constr_prob_f} is not possible.
\revisedText{blue}{We address this by proposing an approximation method with formal error quantification that yields a
tractable and scalable reformulation of~\eqref{eq_smpc}.}

%% file: sections/gmm_unc_prop.tex
\label{sec_mix_prop}
We extend the mixture propagation of~\cite{Adams2025FormalNoise} to approximate $P_{x_{k|t}}$ while
accounting for its dependence on the policy parameters, which enables closed-form Wasserstein bounds
on the approximation error. 

\subsection{Formal Gaussian mixture propagation}

Following \cite{Adams2025FormalNoise}, our approach is to approximate $P_{x_{k|t}}$ with a Gaussian mixture approximation, which we call $\hat{P}_{x_{k|t}}$. To define $\hat{P}_{x_{k|t}}$, for each time step $k$ we consider a set of $m$ points in $\R^{n_x}$ denoted by $\mathcal{C}_{k|t} = \{c_{k|t}^{(i)}\}_{i=1}^m$ and a partition  of $\R^{n_x}$ into $m$ regions called $\mathcal{R}_{k|t} = \{\mathcal{R}_{k|t}^{(i)}\}_{i=1}^m$. How to select these quantities will be discussed in Section \ref{sec_disc_locs}. 
Then, $\hat{P}_{x_{k|t}}$ is defined by the following recursion:
\begin{align}
    \label{eq_mixprop_1}
    \hat{P}_{x_{1|t}} &= T(x_t, \theta_{0|t})\\
    \hat{P}_{x_{k+1|t}} &= \mathbb{E}_{\Delta_{\mathcal{C}_{k|t}, \mathcal{R}_{k|t}}}\left[
    T\left(x_{k|t}, \theta_{k|t}\right)\right] \nonumber
    \\  \label{eq_mixprop_k}
    &=\sum_{i=1}^m \hat{\mathbb{P}}_{x_{k|t}}(\mathcal{R}_{k|t}^{(i)}) \, 
    T\left(c_{k|t}^{(i)}, \theta_{k|t}\right)
    , \quad k\ge 1,
\end{align} 
where 
\(
\Delta_{\mathcal{C}_{k|t}, \mathcal{R}_{k|t}} = \sum_{i=1}^{m} \hat{\mathbb{P}}_{x_{k|t}}(R_{k|t}^{(i)}) \, \delta_{c_{k|t}^{(i)}}
\)
is a discrete distribution approximating $\hat{P}_{x_{k|t}}$. Note that $\Delta_{\mathcal{C}_{k|t}, \mathcal{R}_{k|t}}$ has support $\mathcal{C}_{k|t}$ and probabilities given by the probability mass of $\hat{P}_{x_{k|t}}$ over each of the sets in $\mathcal{R}_{k|t}$. Intuitively, $\hat{P}_{x_{k+1|t}}$ is obtained by recursively discretizing the Gaussian mixture at the previous time steps into discrete distributions of size $m$. Then, because $T$ is a Gaussian mixture distribution,  this guarantees that for every $k$, $\hat{P}_{x_{k|t}}$, as defined in \eqref{eq_mixprop_k},
is also a Gaussian mixture distribution of size\footnote{For simplicity, we also augment $\hat{P}_{x_{1|t}}$ with $M-m_\eta$ extra components by repeating components and distributing the corresponding weights} $M=m\cdot m_{\eta}$, where $m_{\eta}$ is the number of components of $P_{\eta}$.
Furthermore, Proposition 2 in \cite{Adams2025FormalNoise} guarantees that
 for all $k=1, \ldots, N$:
\begin{equation}
    \label{eq_wrho_bound}
    W_2(P_{x_{k|t}}, \hat{P}_{x_{k|t}}) \le \beta_{k|t}\,\,, 
\end{equation}
for some known and tractable $\beta_{k|t} > 0$.
\begin{remark}
\label{remark:ErrorDisccretization}
    $\beta_{k|t}$ quantifies the approximation error after $k$ time steps\footnote{For the analytic derivation of $\beta_{k|t}$ we refer to Proposition 2 in \cite{Adams2025FormalNoise}}. Consequently, it depends on the accuracy of the discretization of the various $\hat{P}_{x_{l|t}}$ for $l<k$, which depend on the choice of $\mathcal{C}_{l|t}$ and $\mathcal{R}_{l|t}$. In particular, for a given $\mathcal{C}_{l|t}$, it is well known that the optimal $\mathcal{R}_{l|t}$ is given by the Voronoi partition of $\R^{n_x}$ w.r.t. $\mathcal{C}_{l|t}$ \cite{Graf2000FoundationsDistributions}. 
\end{remark}

\subsection{Efficient Gaussian mixture discretization} \label{sec_disc_locs}\revisedText{blue}{The Voronoi partition, which leads to optimal construction as stated in Remark \ref{remark:ErrorDisccretization}, when induced by a general $\mathcal{C}_{k|t}$, cannot always be computed
efficiently. Therefore, we follow the approach of~\cite{Adamsdiscretizedistributions2025}, summarized in Figure~\ref{fig_modes_fig}.}
 This approach results in a suboptimal discretization, but also leads to tractable closed forms for the Voronoi partitions and, consequently, for the probability weights in \eqref{eq_mixprop_k} and the Wasserstein bounds in \eqref{eq_wrho_bound} \cite{Adamsdiscretizedistributions2025}.
\begin{figure}[tb]
    \centering
\includegraphics[width=0.45\textwidth]{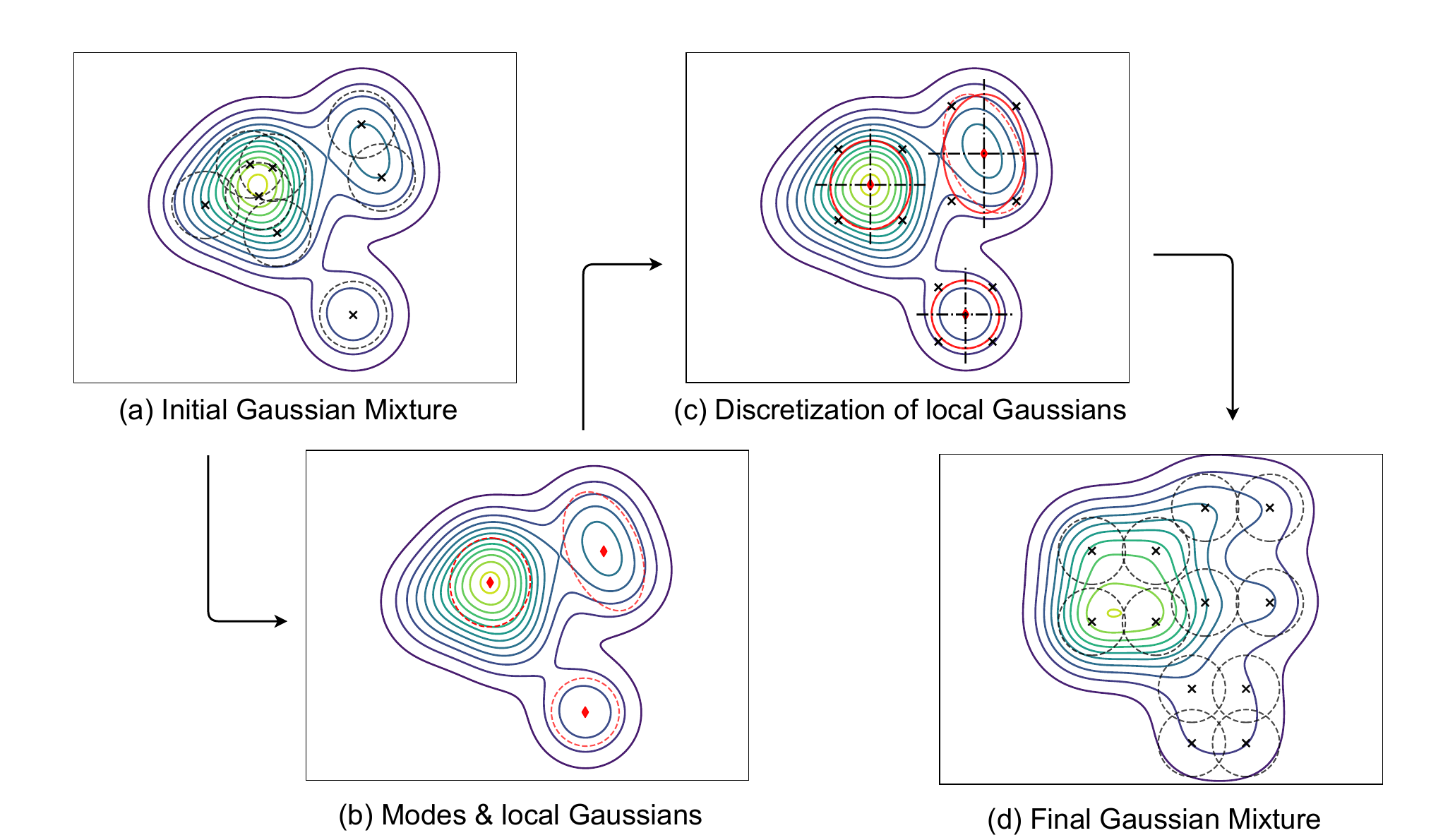}
    \caption{{\hspace{-1.5mm}Mixture propagation for $x^{+}=x+\eta$, $\eta \sim \mathcal{N}(0,0.05)$.
(a) Initial $8$-component mixture with identical covariance (means: black points, circles:
covariance), (b) modes and fitted local Gaussians, (c) $4$-point grids per mode from the orthonormal
projection of the local covariances, weighted by the mixture mass on the induced Voronoi regions,
(d) resulting mixture for $x^{+}$.}
}
    \label{fig_modes_fig}
    \vspace{-6mm}
\end{figure}
The key idea is to identify the \emph{modes} (i.e., local maxima) of $\hat{P}_{x_{k|t}}$, which form the set 
    \(
    \mathcal{S}_{k|t} = \{ s\in\R^{n_x}: \nabla \hat{p}_{x_{k|t}}(s) = 0 \,\wedge\,\nabla^2\hat{p}_{x_{k|t}}(s) \prec 0\}\).
For this, we choose the hybrid mean-shift / Newton algorithm proposed in \cite{Carreira-PerpinMode-FindingDistributions}.\looseness=-1

Next, a \emph{local Gaussian} is fitted around each mode. Using a Laplace approximation, each local Gaussian has a mean equal to the corresponding mode $s_{k|t}^{(i)} \in \mathcal{S}_{k|t}$, and a local covariance $\Sigma_{k|t}^{(i)}$ given by:
\begin{equation}
    \label{eq_local_mode_cov}
    \Sigma_{k|t}^{(i)} = -\left[\left.\nabla^2\log \hat{p}_{x_{k|t}}(x)\right|_{x=s_{k|t}^{(i)}}\right]^{-1}.
\end{equation}

Each local Gaussian is then discretized separately with its assigned budget of $L = m/|\mathcal{S}_{k|t}|$\footnote{We assume for simplicity that $|\mathcal{S}_{k|t}|$ divides $m$ for all $k$, as we can always duplicate or remove modes from the set $\mathcal{S}_{k|t}$ to make this hold}. 
In particular, for every $k$ let $i = 1, \ldots, |\mathcal{S}_{k|t}|$ and $j = 1, \ldots, L$ define an enumeration for the locations $c_{k|t}^{(i, j)} \in \mathcal{C}_{k|t}$. Then, after the discretization of each component takes place and after taking into account the diagonal structure of the covariance matrix $\Sigma_\eta$, the final selection for the locations will be:
\begin{equation}
    \label{eq_gmm_locs_disc}
    c_{k|t}^{(i, j)} = s_{k|t}^{(i)} + \mathrm{diag}\big[(\Sigma_{k|t}^{(i)})^{1/2}\big]\,\odot\, \bar{c}_{k|t}^{(i, j)},
\end{equation}
where $\odot$ is the Hadamard product, $\operatorname{diag}[\cdot]$ is the matrix diagonal, $\bar{c}_{k|t}^{(i, j)}\in \bar{\mathcal{C}}_{k|t}^{(i)}\subseteq \R^{n_x}$ with $|\bar{\mathcal{C}}_{k|t}^{(i)}|=L$ is a set of points that form a rectangular grid centered at zero for each mode $i$, chosen as in \cite{Adams2024FiniteSelection}. This grid choice is such that it enables a tractable closed form expression for the calculation of the Voronoi region probabilities.\looseness=-1 

In the controlled setting, the locations $c^{(i)}_{k|t}$ depend recursively on $\theta_{0:k-1|t}$, and are thus functions of the policy parameters. We
extend~\cite{Adamsdiscretizedistributions2025} to ensure the local smoothness of those functions, by selecting $\bar{\mathcal{C}}_{k|t}^{(i)}$, $\Sigma_{k|t}^{(i)}$ and $|\mathcal{S}_{k|t}|$ (the number of modes considered) constant w.r.t. $\theta_{0:k-1|t}$. Then, the dependence of the locations on $\theta_{0:k-1|t}$ is only via \eqref{eq_gmm_locs_disc} through the modes $s_{k|t}^{(i)}$, and the Implicit Function Theorem (IFT) for $\nabla\hat{p}_{x_{k|t}}(s_{k|t}^{(i)})=0$ implies that the maps $\theta_{0:k-1|t}\rightarrow s_{k|t}^{(i)}$ exist and are locally smooth. This fact enables us to perform gradient-based optimization in Section \ref{sec_imlp_all} while still warm-starting the parameters kept constant from the original propagation scheme of \cite{Adamsdiscretizedistributions2025} to avoid overly conservative bounds $\beta_{k|t}$.\looseness=-1

%% file: sections/mixmpc/cost_and_cdf_derv.tex
\subsection{Dynamics in the Gaussian mixture parameter space}
A Gaussian mixture distribution is fully characterized by its weights, the means and variances of its components. Consequently, to fully characterize $\hat{P}_{x_{k|t}}= \mathrm{GM}(w_{k|t},\, \mu_{k|t},\, \Sigma_\eta) \in \mathcal{G}_M^{\Sigma_\eta}(\R^{n_x})$, we just need to show how the weights and the means of its components evolve over time.
As  $\mu_{k|t} \in (\mathbb{R}^{n_x})^M$
and $w_{k|t} \in \mathcal{W}^M$ are obtained by the convolution of $P_{\eta}$ with a discrete distribution of size $m$ \eqref{eq_mixprop_k}, we use a two-index enumeration\footnote{
We alternate the notation between the double indexing $(i,j)$ for the dynamics and the flattened index $l=1,\ldots, M$ for the cost and constraints.\looseness=-1
} \(w_{k|t} = (w_{k|t}^{(1, 1)}, \ldots, w_{k|t}^{(m, m_\eta)}) \) and \(\mu_{k|t} = (\mu_{k|t}^{(1, 1)}, \ldots, \mu_{k|t}^{(m, m_\eta)})\). Consequently, we can rewrite \eqref{eq_mixprop_1} and \eqref{eq_mixprop_k} as:
\begin{align}
\mu_{1|t}^{(i, j)} &= f_\pi(x_t, \theta_{0|t}) + \mu_\eta^{(j)}, \label{eq_mean1_prop}\\
w_{1|t}^{(i, j)} &= \frac{1}{m} w_\eta^{(j)}, \label{eq_w1_prop}\\
\mu_{k|t}^{(i, j)} &= f_\pi(c_{k|t}^{(i)}, \theta_{k|t}) + \mu_\eta^{(j)}, \label{eq_mk_prop}\\
w_{k|t}^{(i,j)} &= w_\eta^{(j)} \, \hat{\mathbb{P}}_{x_{k|t}}(\mathcal{R}_{k|t}^{(i)})\label{eq_w_k_prop},
\end{align}
with $i=1, \ldots, m$ and $j = 1, \ldots, m_\eta$, where $c_{k|t}^{(i)}$ and $\mathcal{R}_{k|t}^{(i)}$ are selected via the procedure described in Section \ref{sec_mix_prop}-B. 

Finally, we introduce the parameter vector $z_{k|t} := (\mu_{k|t}, w_{k|t}) \in \mathcal{Z}_M := (\mathbb{R}^{n_x})^M \times \mathcal{W}^M$ and the map $\psi:\mathcal{Z}_M \times \Theta \rightarrow \mathcal{Z}_M$, which implements the Gaussian mixture propagation in the induced mixture parameter space, $z_{k+1|t} = \psi(z_{k|t}, \theta_{k|t})$, according to \eqref{eq_mk_prop} and \eqref{eq_w_k_prop}.

\subsection{Expected cost under the Gaussian mixtures} \label{sec_cost_reform}
We propose a tractable approximation for the expected cost $J_{P|t}(\theta_{\cdot|t})$ in \eqref{eq_smpc_obj} using instead the expected cost taken under the derived Gaussian mixture approximations:
\begin{equation}
\begin{aligned}
\label{eq_approx_cost}
J_{\hat{P}|t}(\theta_{\cdot|t}) \coloneq  \sum_{k=1}^{N} \mathbb{E}_{\hat{P}_{x_{k|t}}}[\ell_k(x_{k|t})],
\end{aligned} 
\end{equation}
with the simplified notation $\ell_k(x_{k|t}) \coloneq \ell_\pi(x_{k|t}, \theta_{k|t})$ for $k=1, \ldots, N-1$. By the linearity of the expectation operator we can work separately with each Gaussian component as:
\begin{equation}
\begin{aligned}
    \label{eq_intermediate_obj_2}
    J_{\hat{P}|t}(\theta_{\cdot\mid t}) = \sum_{k=1}^{N}\sum_{l=1}^M w_{k|t}^{(l)} \, \mathbb{E}_{\mathcal{N}(\mu_{k|t}^{(l)}, \Sigma_\eta)}[\ell_k(x_{k|t})]
\end{aligned}    
\end{equation}
Following common approaches for Gaussian expectations (e.g. \cite{Scampicchio2025GaussianControl}), we use a first-order Taylor approximation w.r.t. $x_{k|t}$ for $\ell_k(x_{k|t})$ around each component's mean:
\begin{equation}
\label{eq_taylor_obj}
\begin{aligned}
\ell_k(x_{k|t})&\approx \ell_k(\mu_{\revisedText{blue}{k}|t}^{(l)}) + \nabla \ell_k(\mu_{\revisedText{blue}{k}|t}^{(l)})^T(x_{k|t} - \mu_{k|t}^{(l)})
\end{aligned}
\end{equation}
The second term results to zero expectation under ${\mathcal{N}(\mu_{k|t}^{(l)}, \Sigma_\eta)}$. Thus, the expectations in \eqref{eq_intermediate_obj_2} are simplified to:
\begin{equation}
\label{eq_cost_deriv_per_comp}
\begin{aligned}
    \mathbb{E}_{\mathcal{N}(\mu_{k|t}^{(l)}, \Sigma_\eta)}\left[\ell_k(x_{k|t})\right] &\approx \ell_k(\mu_{k|t}^{(l)}).
\end{aligned}
\end{equation}
After substitution of \eqref{eq_cost_deriv_per_comp} in \eqref{eq_intermediate_obj_2} we finally get the approximation $\hat{J}_{\hat{P}|t}(\theta_{\cdot|t}) \approx J_{\hat{P}|t}(\theta_{\cdot|t})$, given by: 
\begin{equation}
\begin{aligned}
\hat{J}_{\hat{P}|t}(\theta_{\cdot|t}) &\coloneq\sum_{k=1}^{N} \sum_{l=1}^M w_{k|t}^{(l)} \, \ell_k(\mu_{k|t}^{(l)}) \\&= \sum_{l=1}^M\sum_{k=1}^{N-1}w_{k|t}^{(l)} \, \ell_\pi(\mu_{k|t}^{(l)}, \theta_{k|t}) + w_{N|t}^{(l)}\,\ell_N(\mu_{N|t}^{(l)}) 
\end{aligned}
   \label{eq_hatJ_main}
\end{equation}
Regarding the quality of the approximation, we provide the following property for the case of quadratic cost functions:
\begin{proposition}[Quadratic cost approximation correctness] \label{thm_cost_dec}
\label{prop_quad_cost_bound}
Assume the stage and terminal costs \revisedText{blue}{under the policy $\pi$} are quadratic in $x$ with
\(
\ell_k(x_{k|t}) =  \|x_{k|t} - r_{k|t}\|_{Q_{k}}^2 + b_{k|t}
\) 
for $b_{k|t} \in \R$, $r_{k|t} \in \R^{n_x}$ and $Q_k \in \mathbb{S}_{++}^{n_x}$, $k=1, \ldots, N$. Then:
\begin{equation}
    J_{\hat{P}|t}(\theta_{\cdot|t}) = \hat{J}_{\hat{P}|t}(\theta_{\cdot|t}) +\revisedText{blue}{\sum_{k=1}^{N}\mathrm{Tr}(Q_k \Sigma_\eta)}. \label{thm1_eq1}
\end{equation}
Also, the deviation from $J_{P|t}(\theta_{\cdot|t})$ will be bounded by:
\begin{equation}
\begin{aligned}
\big|J_{\hat{P}|t}(\theta_{\cdot|t}) - J_{P|t}(\theta_{\cdot|t})\big| \le
\,\sum_{k=1}^{N}\,\|Q_k\|\, \bar{r}_{k|t} \, \beta_{k|t},
\end{aligned}
\label{eq_cost_cost_bound}
\end{equation}
where \(\bar{r}_{k|t} = \sqrt{2\big(\mathbb{E}_{\hat{P}_{x_{k|t}}}\![\lVert x - r_{k|t} \rVert^2] + \mathbb{E}_{P_{x_{k|t}}}\![\lVert x - r_{k|t} \rVert^2] \big)}\).
\end{proposition}
\vspace{0.5mm}
\begin{proof}
The expectation of the quadratic costs over a Gaussian distribution has a closed form given by
\begin{equation}
\mathbb{E}_{\mathcal{N}(\mu_{k|t}^{(l)},\Sigma_\eta)}[\ell_k(x_{k|t})] = \ell_k(\mu_{k|t}^{(l)}) + \mathrm{Tr}(\revisedText{blue}{Q_k} \Sigma_\eta).
\end{equation}
Summing over all components of each Gaussian mixture in the horizon, and then over the horizon $N$ results in \eqref{thm1_eq1}.
Now, applying the triangle inequality to the LHS of \eqref{eq_cost_cost_bound} gives: 
\begin{align}
&|J_{\hat{P}|t}(\theta_{\cdot|t}) - J_{P|t}(\theta_{\cdot|t})| \le 
\label{eq_JJ_proof_step_trian} \\ &\sum_{k=1}^{N} \big|\,\mathbb{E}_{\hat{P}_{x_{k|t}}}[\|x_{k|t} - r_{k|t}\|_{Q_k}^2] - \mathbb{E}_{P_{x_{k|t}}}[\|x_{k|t} - r_{k|t}\|_{Q_k}^2]\,\big| \nonumber
\end{align}
For any $x\in\R^{n_x}, y\in\R^{n_x}$ and any joint distribution $\gamma\in\mathcal{P}_2(\R^{n_x} \times \R^{n_x})$ we have
\begin{flalign}
&\mathbb{E}_\gamma[|x^TQx - y^TQy|] = \mathbb{E}_\gamma[|(x - y)^TQ(x + y)|] \le \nonumber\\ 
&\|Q\|\, (2\mathbb{E}_\gamma[\|x\|^2+\|y\|^2])^{1/2}\, (\mathbb{E}_\gamma[\|x-y\|^2])^{1/2} \label{eq_inter_proof_new_1}
\end{flalign}
due to the Cauchy-Schwarz and $\|x+y\|^2 \le 2(\|x\|^2+\|y\|^2)$ inequalities. Applying the triangle inequality and \eqref{eq_inter_proof_new_1} to each term on the RHS of \eqref{eq_JJ_proof_step_trian}, and selecting $\gamma$ as the $2$-Wasserstein optimal coupling of $P_{x_{k|t}}, \hat{P}_{x_{k|t}}$, we get:
\begin{equation}
|J_{\hat{P}|t}(\theta_{\cdot|t}) - J_{P|t}(\theta_{\cdot|t})| \le \sum_{k=1}^N \|Q_k\|\, \bar{r}_{k|t} \,W_2(\hat{P}_{x_{k|t}}, P_{x_{k|t}}) \nonumber
\end{equation}
The result \eqref{eq_cost_cost_bound} finally follows from the bounds \eqref{eq_wrho_bound}.
\end{proof}
\subsection{Chance constraints reformulation using the Gaussian mixture distributions} \label{sec_cc_reform}
Let $h_k(x_{k|t}) \coloneq h_\pi(x_{k|t}, \theta_{k|t})$, for $k=1, \ldots, N-1$ to make the notation more compact.  
We replace the original chance constraints $\mathbb{P}_{x_{k|t}}\left(h_k(x_{k|t}) \le 0\right)$ of \eqref{eq_smpc} with $\hat{\mathbb{P}}_{x_{k|t}}\left(h_k(x_{k|t}) \le 0\right)$, which are the constraint satisfaction probabilities calculated under the derived Gaussian mixture distributions. 
As these admit no closed form for nonlinear $h_k$, we proceed as for the cost
and first apply the law of total probability, for all $k=1,\dots,N$:
\begin{align}
    \hat{\mathbb{P}}_{x_{k|t}}\Big(h_k\big(x_{k|t}\big) \le 0\Big) =
    \sum_{l=1}^M w_{k|t}^{(l)} \, \mathcal{N}\Big(h_k\big(x_{k|t}\big) \le 0; \mu_{k|t}^{(l)}, \Sigma_\eta&\Big) 
    \nonumber
\end{align}
Then, a first-order Taylor approximation of $h_k(\cdot)$ around each component's mean yields:
\begin{align}
\label{eq_taylor_cc}
h_k(x_{k|t}) &\approx h_k(\mu_{k|t}^{(l)})\, +\nabla h_k(\mu_{k|t}^{(l)})^T (x_{k|t} - \mu_{k|t}^{(l)})
\end{align}
Notice that each affine approximation is normally distributed with mean $h_k(\mu_{k|t}^{(l)})$ and variance $\|\nabla h_k(\mu_{k|t}^{(l)})\|^2_{\Sigma_\eta}$. Therefore, using the standard normal CDF $\Phi(\cdot)$, we end up with the following closed form approximation:
\begin{equation}
\hat{\mathbb{P}}_{x_{k|t}}\big(h_k\big(x_{k|t}\big) \le 0\big) \approx \sum_{l=1}^M w_{k|t}^{(l)}\Phi\Big( -\frac{h_k\big(\mu_{k|t}^{(l)}\big)}{ \big\| \nabla h_k(\mu_{k|t}^{(l)}) \big\|_{\Sigma_\eta}\,}\Big)\label{eq_CDF_approx_f}
\end{equation}
\revisedText{blue}{Note that without curvature information, \eqref{eq_CDF_approx_f} is a heuristic approximation, whereas for affine $h_k$ it becomes an exact reformulation\footnote{The RHS of~\eqref{eq_CDF_approx_f} is well defined even when the denominator vanishes, as long as $h_k(\mu_{k|t}^{(l)}) \neq 0$, since $\Phi(+\infty) =1$ and $\Phi(-\infty) =0$; 
in practice, we add a small positive constant to the denominator.}.}
To characterize the error introduced from the approximation \eqref{eq_CDF_approx_f} we provide the following result for the simple but very informative case of affine functions $h_k$:
\begin{proposition}[Affine constraints guarantees] \label{thm2}
    Suppose that the constraint functions \revisedText{blue}{under the policy $\pi$} are affine in $x$, with $\nabla h_k(x_{k|t}) = a_{k|t} \in \R^{n_x}$, for non-zero $a_{k|t} \in \R^{n_x}$ and all $k=1,\ldots, N$. Then the approximate equality \eqref{eq_CDF_approx_f} becomes an exact equality. Furthermore, the squared difference of the true and approximate probabilities satisfies:
    \begin{align} 
        \big|\mathbb{P}_{x_{k|t}}\left( h_k(x_{k|t})\le 0\right) - \hat{\mathbb{P}}_{x_{k|t}}\left(h_k(x_{k|t})\le 0\right)\big|^2 \nonumber \\\le  \sqrt{2/\pi}\, \|a_{k|t}\|_{\Sigma_\eta}^{-1}\, \|a_{k|t}\|\, \beta_{k|t}  .\label{eq_prob_bound}
    \end{align}
\end{proposition}
\vspace{1mm}
\begin{proof}
    For affine functions the first-order Taylor approximation is exact and thus the first part of the proposition is straightforward. From \cite[prop. 1.2]{ross2011fundamentalssteinsmethod} we know that
    \begin{align} 
        \big|\mathbb{P}_{x_{k|t}}\left( h_k(x_{k|t})\le 0\right) - \hat{\mathbb{P}}_{x_{k|t}}\left(h_k(x_{k|t})\le 0\right)\big|^2 \nonumber \\\le 2\,\bar{p}_{\max}({h_k}_\#\hat{P}_{x_{k|t}})\, W_1({h_k}_\#\hat{P}_{x_{k|t}}, {h_k}_\#P_{x_{k|t}}),  \label{eq_cc_interm_2}
    \end{align}
    where $\bar{p}_{\max}({h_k}_\#\hat{P}_{x_{k|t}})$ is the maximum value of the pdf of the push-forward distribution ${h_k}_\#\hat{P}_{x_{k|t}}$. Since $h_k$ is affine, ${h_k}_\#\hat{P}_{x_{k|t}}$ will be a Gaussian mixture with shifted means and shared across components variance equal to $\|a_{k|t}\|_{\Sigma_\eta}^2$. Therefore, the maximum value of the underlying pdf will satisfy $\bar{p}_{\max}({h_k}_\#\hat{P}_{x_{k|t}}) \le(\sqrt{2\pi}\,\|a_{k|t}\|_{\Sigma_\eta})^{-1}$. The result follows via the bounds \eqref{eq_wrho_bound} after applying in \eqref{eq_cc_interm_2} successively the well-known inequalities \cite{CedricVillaniOptimalNew}:
    \begin{align}
         W_1({h_k}_\#\hat{P}_{x_{k|t}}, {h_k}_\#P_{x_{k|t}}) &\le \|a_{k|t}\| W_1(\hat{P}_{x_{k|t}}, P_{x_{k|t}})\\ &\le \|a_{k|t}\| W_2(\hat{P}_{x_{k|t}}, P_{x_{k|t}}) ,
    \end{align}
    for affine $h_k$ with $\nabla h_k(x_{k|t}) = a_{k|t}$.
\end{proof}

%% file: sections/mixmpc/formulation.tex
The complete resulting optimization problem that approximates the original SNMPC \eqref{eq_smpc}, using the derived Gaussian mixture distributions over the predicted state, is given by:
\begin{mini!}|s|[2]
{\substack{\theta_{\cdot|t} \in \Theta^N
\\ z_{\cdot|t}\in \mathcal{Z}_M^N}}
{\sum_{l=1}^{M} \sum_{k=1}^{N-1} w_{k|t}^{(l)} \ell_\pi \big(\mu_{k|t}^{(l)}, \theta_{k|t}\big) 
+ w_{N|t}^{(l)}\ell_N\big(\mu_{N|t}^{(l)}\big) \label{eq_mixmpc_objective}}
{\label{eq_mixmpc}}
{}
\addConstraint{\mu_{1|t}^{(i, j)} = f_\pi(x_t, \theta_{0|t}) + \mu_\eta^{(j)} \label{eq_mu1_mixmpc}}
\addConstraint{w_{1|t}^{(i, j)} = \frac{1}{m} w_\eta^{(j)} \label{eq_w1_mixmpc}}
\addConstraint{z_{k+1|t} = \psi(z_{k|t}, \theta_{k|t}) \label{eq_mixprop_mixmpc}}
\addConstraint{\sum_{l=1}^M w_{k|t}^{(l)}}{\,\Phi\Big( \frac{h_\pi\big(\mu_{k|t}^{(l)}, \theta_{k|t}\big)}{\big\| \nabla h_\pi(\mu_{k|t}^{(l)},\theta_{k|t}) \big\|_{\Sigma_\eta}}\Big)\le \bar{\epsilon}_k \label{eq_mixmpc_cdf_k}}
\addConstraint{\sum_{l=1}^M w_{N|t}^{(l)}}{\,\Phi\Big( \frac{h_N\big(\mu_{N|t}^{(l)}\big)}{\big\| \nabla h_N(\mu_{N|t}^{(l)}) \big\|_{\Sigma_\eta}}\Big)\le \bar{\epsilon}_N \label{eq_mixmpc_cdf_N}}
\end{mini!}

Equation \eqref{eq_mixmpc_objective} is the objective function as derived in Section \ref{sec_cost_reform}. The constraints \eqref{eq_mu1_mixmpc}--\eqref{eq_w1_mixmpc} initialize the first mixture analogously to \eqref{eq_mean1_prop}--\eqref{eq_w1_prop}, and~\eqref{eq_mixprop_mixmpc} preserves the
propagation mechanism of Section~\ref{sec_mix_prop} for all $k=1,\dots,N-1$. Finally, the constraints \eqref{eq_mixmpc_cdf_k}, \eqref{eq_mixmpc_cdf_N} \revisedText{blue}{are expressed via the Taylor approximation of Section \ref{sec_cc_reform}, where we used the property $\Phi(-y) = 1 - \Phi(y)$ to translate them to violation probabilities,} and enforced with possibly tightened probabilities $1 - \bar{\epsilon}_k \ge 1 - \epsilon_k$, for $k=1, \ldots,N$.
To provide a formal notion of correctness guarantees for the approximate optimization problem, let us consider again quadratic-affine problems. Then, we can control the quality of the approximation to the original SNMPC via a proper selection of $M$ and $\bar{\epsilon}_k$, based on the following theorem:  
\begin{theorem}[Correctness and optimality guarantees] \label{thm_conv}
Assume that $f_\pi$ is globally Lipschitz continuous, \revisedText{blue}{and that under the policy $\pi$, $h_\pi, h_N$ are affine and $\ell_\pi, \ell_N$ quadratic} in $x$ as in Propositions \ref{thm_cost_dec}--\ref{thm2}.
Then, for any choice of $\bar{\epsilon}_k$, with $\bar{\epsilon}_k < \epsilon_k$, $k=1,\ldots, N$, 
there exists an $\bar{M} \in \mathbb{N}_+$ such that any feasible solution of \eqref{eq_mixmpc} with $M\ge\bar{M}$ is also feasible for \eqref{eq_smpc}. Furthermore, the optimal solution of \eqref{eq_mixmpc}, denoted by $\tilde{\theta}_{\cdot|t}$, satisfies:\looseness=-1 \vspace{-2.5mm}
\begin{align}
    J_{P|t}(\tilde{\theta}_{\cdot|t}) - J_{P|t}(\theta^*_{\cdot|t}) \le 2\sum_{k=1}^{N}\,\|Q_k\|\, \bar{r}_{k|t, \max}\,\beta_{{k|t},\max} \label{opti_bound_eq}
\end{align}
with $\bar{r}_{k|t, \max} = \max_{\tilde{\theta}_{\cdot|t}, \theta^*_{\cdot|t}} \bar{r}_{k|t}$, for the $\bar{r}_{k|t}$ of Proposition \ref{thm_cost_dec}, evaluated when the mixture distributions are derived with each one of $\tilde{\theta}_{\cdot|t}, \theta^*_{\cdot|t}$, and $\beta_{{k|t}, \max}$ is defined analogously.\looseness=-1
\end{theorem}
\begin{proof}
For globally Lipschitz dynamics, Proposition 3 of \cite{Adams2025FormalNoise} implies that for the case of no distributional uncertainty, the bounds $\beta_{k|t}$ can be made arbitrarily small by reducing the discretization error of the mixture distributions. As this is controlled by the number of components in the mixtures for our controlled setting too, we can always find $\bar{M}\in\mathbb{N}_+$ such that for all $M\ge \bar{M}$, we have $\beta_{k|t} \le  \sqrt{\pi/2}\,\|a_{k|t}\|_{\Sigma_\eta}\,\|a_{k|t}\|^{-1}(\epsilon_{k} - \bar{\epsilon}_{k})^2$. This implies feasibility of the nominal chance constraints through Proposition \ref{thm2}, as:
\begin{align*}
    \mathbb{P}_{x_{k|t}}(h_k(x_{k|t})\le 0) &\ge \hat{\mathbb{P}}_{x_{k|t}}(h_k(x_{k|t})\le0)\\ &\quad- (\sqrt{2/\pi}\,\|a_{k|t}\|_{\Sigma_\eta}^{-1}\,\|a_{k|t}\|\,\beta_{k|t})^{1/2} \\&
    \ge  1 - \bar{\epsilon}_k - ((\epsilon_k - \bar{\epsilon}_k)^2)^{1/2} = 1 -\epsilon_k
\end{align*}
For the sub-optimality bound \eqref{opti_bound_eq}, take
\begin{align}
\label{eq_interm_proof_thm3}
J_{P|t}(\tilde{\theta}_{\cdot|t}) - J_{P|t}(\theta^*_{\cdot|t}) = \left[\hat{J}_{\hat{P}|t}(\theta^*_{\cdot|t}) - J_{P|t}(\theta^*_{\cdot|t}) \right] + \nonumber\\ \left[J_{P|t}(\tilde{\theta}_{\cdot|t}) - \hat{J}_{\hat{P}|t}(\tilde{\theta}_{\cdot|t}) \right] + \left[\hat{J}_{\hat{P}|t}(\tilde{\theta}_{\cdot|t})- \hat{J}_{\hat{P}|t}(\theta^*_{\cdot|t})\right]
\end{align}
The last term is non-positive as $\tilde{\theta}_{\cdot|t} = \arg\min\hat{J}_{\hat{P}|t}(\theta_{\cdot|t})$. 
From Proposition \ref{thm_cost_dec} for quadratic costs, we have that:
\[
\hat{J}_{\hat{P}|t}(\theta^*_{\cdot|t})- \hat{J}_{\hat{P}|t}(\tilde{\theta}_{\cdot|t}) = J_{\hat{P}|t}(\theta^*_{\cdot|t})- J_{\hat{P}|t}(\tilde{\theta}_{\cdot|t})
\]
The triangle inequality in \eqref{eq_interm_proof_thm3} then yields: \vspace{-1mm}
\begin{align*}
J_{P|t}(\tilde{\theta}_{\cdot|t}) - J_{P|t}(\theta^*_{\cdot|t}) \le \left|J_{\hat{P}|t}(\theta^*_{\cdot|t}) - J_{P|t}(\theta^*_{\cdot|t}) \right| +\\ \left|J_{P|t}(\tilde{\theta}_{\cdot|t}) - J_{\hat{P}|t}(\tilde{\theta}_{\cdot|t}) \right|
\end{align*}
Applying Proposition \ref{thm_cost_dec} to both terms and taking the stage-wise maxima of $\bar{r}_{k|t}$ and $\beta_{k|t}$ results in \eqref{opti_bound_eq}.
\end{proof}
\begin{corollary}[Asymptotic convergence]
As $M\to\infty$, the sub-optimality gap of \eqref{opti_bound_eq} asymptotically goes to zero. \label{cor_1}
\end{corollary}
\begin{proof}
    As stated in the proof of Theorem \ref{thm_conv}, we can always find a sufficiently large $M$ such that $\beta_{k|t} \le \sqrt{\pi/2}\,\|a_{k|t}\|_{\Sigma_\eta}\,\|a_{k|t}\|^{-1}(\epsilon_{k} - \bar{\epsilon}_{k})^2$.
    Letting $M\rightarrow\infty$ enables us also to choose $\bar{\epsilon}_k$ arbitrarily close to $\epsilon_k$, thus taking the limit on the above expression as $\bar{\epsilon}_k\rightarrow\epsilon_k$ yields $\beta_{k|t}\rightarrow0$. The result follows after letting $\beta_{k|t}\rightarrow0$ in \eqref{opti_bound_eq}.    
\end{proof}

\revisedText{blue}{We stress that Theorem \ref{thm_conv} and Corollary \ref{cor_1} concern the open-loop problem, and do not by themselves imply recursive feasibility, closed-loop chance-constraint satisfaction, or stochastic stability of the closed-loop receding-horizon implementation. We leave a formal analysis of these for future work, and highlight that they can be recovered by designing terminal components and a constraint tightening \cite{9044326}.}\looseness=-1

%% file: sections/results/implementation.tex
We solve \eqref{eq_mixmpc} using a Sequential Quadratic Programming (SQP) approach. Before solving the optimization problem at time $t$, we use the optimal parameter sequence of $t-1$ to obtain all Gaussian mixtures over the predicted state in the horizon exactly as in Section \ref{sec_mix_prop}. To avoid the non-smooth operations of the mixture propagation inside the optimization, we then keep fixed the number of modes, the local covariances and the grids $\bar{\mathcal{C}}_{k|t}^{(i)}$ inside all SQP iterations as explained in Section \ref{sec_mix_prop}-B. \revisedText{blue}{Note that this choice does not break the theoretical results presented above; it only has an effect in increasing the Wasserstein error-radii $\beta_{k|t}$ and thus making the bounds more conservative.} We finally introduce extra variables for the modes, $s_{k|t}^{(i)}$, along with the constraints $\nabla\hat{p}_{x_{k|t}}(s_{k|t}^{(i)})=0$ that transform \eqref{eq_mk_prop} to a discrete implicit DAE.
The mixture propagation is implemented in \emph{PyTorch} \cite{paszke2017automatic}, leveraging the package provided by \cite{Adamsdiscretizedistributions2025}. We use \emph{CasADi} \cite{Andersson2019} and \emph{acados} \cite{Verschueren2021} to symbolically formulate and solve the resulting SQP. 
In our current implementation, we also keep the weights constant to further reduce the computational complexity \revisedText{blue}{and highlight again that this affects our results only via an increase in the Wasserstein approximation-error bounds $\beta_{k|t}$.} \revisedText{blue}{The NLP carries only $N$ chance constraints  independent of $M$. Per SQP iteration, the propagation requires $N\cdot m$ evaluations of $f_\pi$ and its Jacobian, while the work to derive the Laplace covariances \eqref{eq_local_mode_cov} and complete the mode-search scales as $\sim n_x^3$ and $\sim n_x^2$ respectively, and linearly in the number of components $M$, the number of modes, and the horizon $N$. Thus the dependence on $n_x$ is the limiting factor, since the per-mode rectangular grids also grow exponentially in $n_x$ (cf. Remark 1).}
\revisedText{blue}{We provide the code and reproducibility details at \url{https://github.com/kprattis/mixmpc}.}


%% file: sections/results/simulations.tex
We evaluate our method with $m=3$ on a set of two nonlinear benchmarks, both for the cases of unimodal Gaussian and bimodal mixture disturbance distributions. We also compare our method with 3 standard MPC variants: a) a Nominal MPC that ignores the disturbance, b) a Stochastic MPC that propagates only the first two moments to derive a Gaussian approximation for the predicted state distribution \cite{Kabzan2019Learning-basedRacing}, and c) a Scenario-based MPC \cite{SCHILDBACH20143009} with $N_s=20$ scenarios.
We examine all control algorithms on the same set of $20$ disturbance realizations for each task. All experiments were performed on an Intel Core i7-1265U CPU with 8 GB RAM.\looseness=-1
\subsubsection{Double vortex 2-D system}
We first consider the two dimensional system with dynamics 
\begin{equation}
\dot{x}  = \begin{cases}
    A \, (x - p_1) + u, &x_1\ge 0 \\
    A^T (x - p_2) + u, &x_1 < 0
\end{cases}    
\end{equation}
for $A=-0.4\,\sqrt{2}\,\begin{bmatrix}1,1;-1,1\end{bmatrix}$, $p_1 = [1.0, 0.0]^T$ and $p_2=[-0.5, -0.25]^T$. In this example, we aim to minimize the cost \(\sum_{k=1}^N \|x_{k|t} - p_1\|_Q^2 + \|u_{k|t}\|_R^2,\) with $Q=I_2$ and $R=0.8\,I_2$, while satisfying constraints of the form $\|x - c_i\|_{E_i}^2 \ge 1$, for $i=1, 2$. In other words, we demand the state to remain outside of two elliptic obstacles centered at $c_2=p_2, \,c_1=[0.5,-0.6]^T$ determined by $E_1, E_2 \in \mathbb{S}_{++}^2$. We then discretize the system with RK4 ($dt=0.1$) and add the discrete disturbance $\eta_t$, distributed according to $\mathcal{N}(0_2, \Sigma_\eta)$ for the unimodal and $0.5\,\mathcal{N}([0.05, 0]^T, \Sigma_\eta) + 0.5\,\mathcal{N}([-0.05,0]^T, \Sigma_\eta)$ for the bimodal case respectively, with $\Sigma_\eta=10^{-3}I_2$. Finally, we use $u_t = \pi(x_t, \theta_t)= \theta_t$ with $|\theta_t|\le \mathbf{1}_{n_u}$, set the horizon to $N=10$ and the satisfaction probability to $1-\epsilon_k = 0.99$ for all $k$. \revisedText{blue}{We highlight that since this example has non-affine constraints (and in fact non-convex), it demonstrates the practical value of our proposed methodology also in a setting where the assumptions of Theorem \ref{thm_conv} are violated, as they would require affine constraints.}
\begin{table}[b!]
\vspace{-5mm}
\renewcommand{\arraystretch}{1.2}
\centering
\caption{Results for the double vortex benchmark}
\label{tab:mpc_results_dv}
\setlength{\tabcolsep}{2pt} 
\begin{tabular}{@{} c c c c @{}}
\toprule
\textbf{MPC-Type} & \textbf{Cost} & \textbf{Max. Viol.} & \textbf{Total/QP (ms)} \\
\midrule
\multicolumn{4}{l}{\textit{Unimodal Disturbance}} \\
Nominal           & -     &  $0.75 \pm 0.23$          & $1.2/0.5$ \\
Gaussian          & $46.6 \pm 4.2$    &   $0.0 \pm 0.0$          & $3.8/2.6$ \\
Scenario          & $29.8 \pm 2.67$     & $0.12 \pm 0.29$          & $156/146$ \\
\textbf{Proposed} & $\mathbf{29.6 \pm 2.15}$ & $\mathbf{0.0 \pm 0.0}$ & $\mathbf{99.7/14.5}$ \\
\midrule
\multicolumn{4}{l}{\textit{Bimodal Disturbance}} \\
Nominal           & -    & $0.67 \pm 0.25$          & $0.9/0.3$ \\
Gaussian          & $48.1 \pm 5.0$    & $0.29 \pm 0.38$          & $6.9/6.9$ \\
Scenario          & $42.6 \pm 4.8$    & $0.17 \pm 0.3$          & $136/135$ \\
\textbf{Proposed}  & $\mathbf{44.9 \pm 13}$  & $\mathbf{0.0 \pm 0.0}$ & $\mathbf{249/51}$ \\
\bottomrule
\end{tabular}
\end{table}
\begin{figure}[t!]
    \centering
    \vspace{-5mm}
\includegraphics[width=0.48\textwidth]{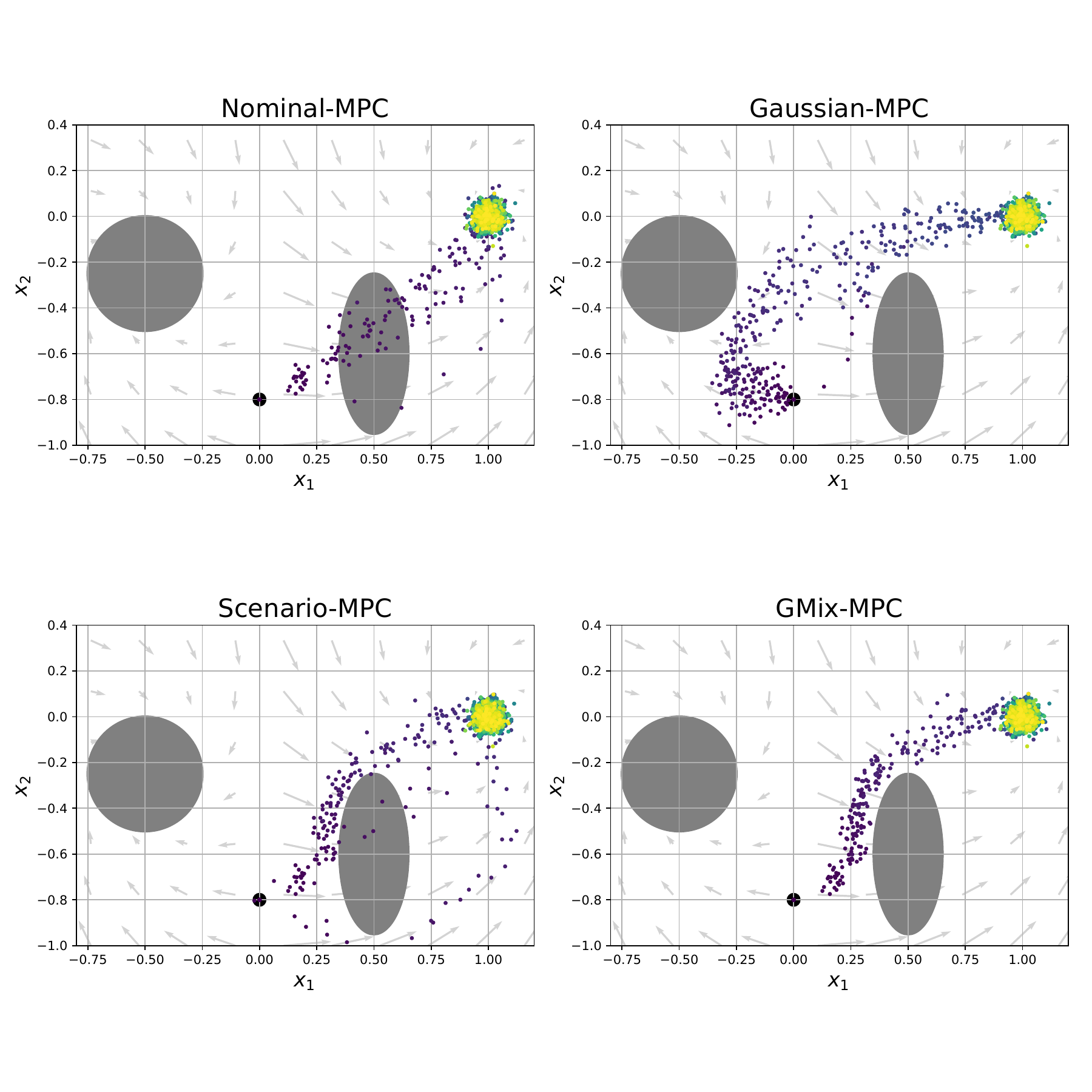}
    \vspace{-13mm}
    \caption{{Double vortex, unimodal noise. Gaussian propagation is overly conservative and forces
trajectories into the left half plane, while ours stays in the right half plane (see Table~\ref{tab:mpc_results_dv}).}}
    \label{fig_dv_results}
    \vspace{-8mm}
\end{figure}
\begin{figure}[t!]
    \centering
\includegraphics[width=0.45\textwidth]{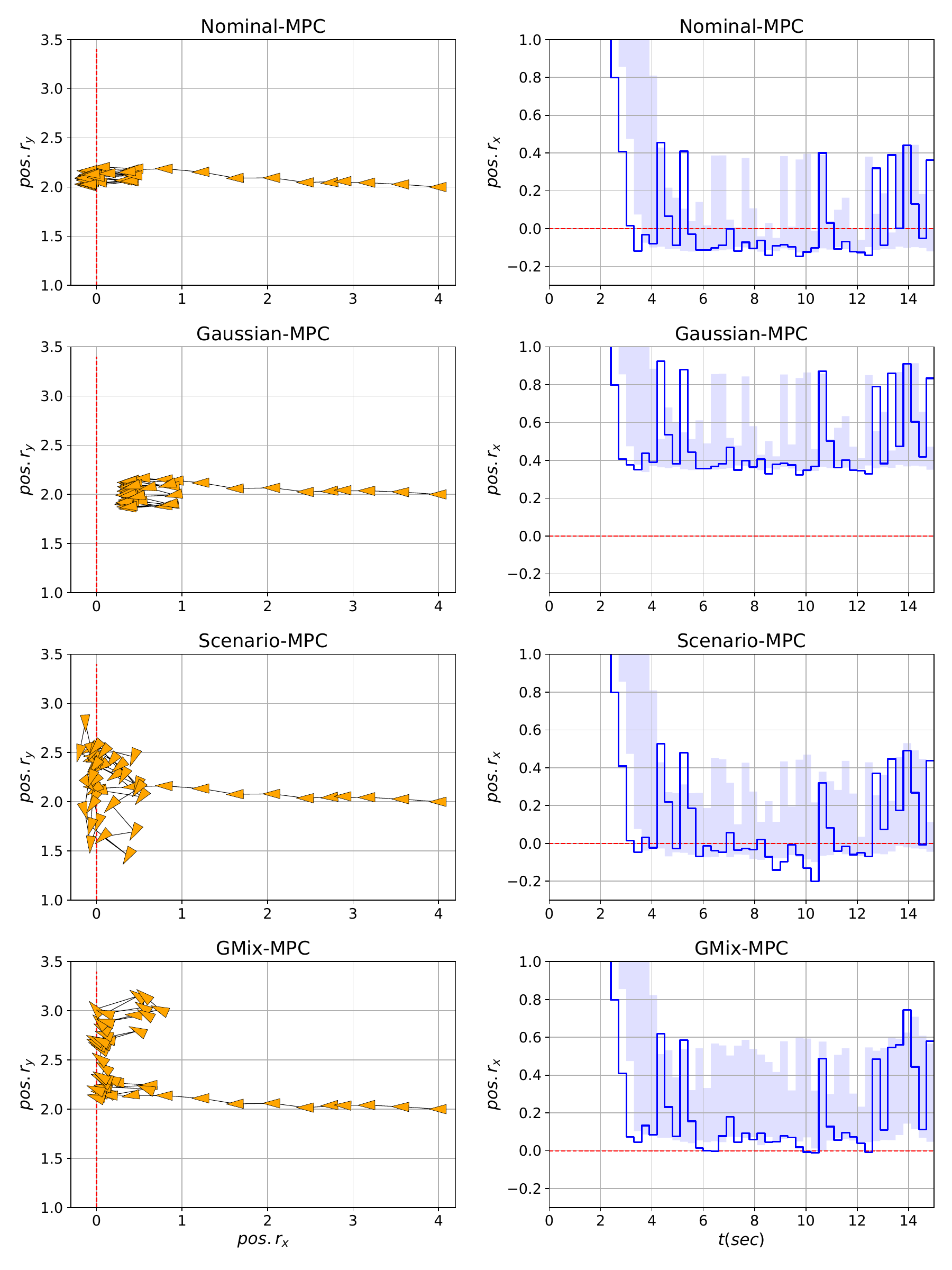}
    \vspace{-3mm}
    \caption{{\hspace{-0.7em}
    Planar robot, bimodal noise. Left: indicative trajectory per controller. Right: distance
from $r_x=0$, with the $80$--$20$ quantile in light blue.
    }
    }
    \vspace{-3mm}
    \label{fig_pr_res}
\end{figure}
\begin{table}[t!]
\renewcommand{\arraystretch}{1.2}
\centering
\caption{Results for the planar robot benchmark}
\label{tab:mpc_results_pr}
\setlength{\tabcolsep}{2pt} 
\begin{tabular}{@{} c c c c @{}}
\toprule
\textbf{MPC-Type} & \textbf{Cost} & \textbf{Max. Viol.} & \textbf{Total/QP (ms)} \\
\midrule
\multicolumn{4}{l}{\textit{Unimodal Disturbance}} \\
Nominal           & -   & $0.06 \pm 0.01$          & $6.2/3.6$ \\
Gaussian          & $79.9 \pm 2.9$    & $0.00 \pm 0.00$          & $15.1/9.7$ \\
Scenario         & $80.8 \pm 2.41$     & $0.1 \pm 0.01$          & $549.6/446.1$ \\
\textbf{Proposed} & $\mathbf{80.0 \pm 2.9}$ & $\mathbf{0.00 \pm 0.00}$ & $\mathbf{86.4/11.8}$ \\
\midrule
\multicolumn{4}{l}{\textit{Bimodal Disturbance}} \\
Nominal           & -    & $0.16 \pm 0.02$          & $7.4/4.3$ \\
Gaussian          & $92.0 \pm 12.8$   & $0.00 \pm 0.00$          & $27.1/16.2$ \\
Scenario         & -     & $0.16 \pm 0.07$          & $665.3/581.8$ \\
\textbf{Proposed} & $\mathbf{86.8 \pm 10.0}$ & $\mathbf{0.07 \pm 0.18}$ & $\mathbf{482.8/207.0}$ \\
\bottomrule
\end{tabular}
\vspace{-4mm}
\end{table}
The results are summarized in Figure \ref{fig_dv_results} and Table \ref{tab:mpc_results_dv}, where we provide the average cost over the valid trajectories (or '-' when all trajectories are invalid) and the average maximum violation with their standard deviations, as well as average total and QP solver times. We can see that our proposed MPC framework (GMix-MPC) offers a practical trade-off between cost optimality, constraint satisfaction, and computational time. When compared to the Gaussian MPC variant, the controller avoids overly conservative actions while remaining safe, due to its more accurate uncertainty propagation, as can be seen in Figure \ref{fig_dv_results}. Table \ref{tab:mpc_results_dv} also confirms that our proposed controller achieves comparable performance to the sampling-based MPC only using $m=3$ locations to form the mixtures, while the Nominal MPC fails to provide a valid trajectory due to its ignorance of the underlying stochasticity. Note that we use slacks for the constraints to allow the experiment's continuation when a trajectory becomes invalid.
\subsubsection{Planar robot 3-D system}
We also consider the planar robot system governed by the Dubins car dynamics studied in \cite{9993720}. The goal is to drive the robot with state $x=[r_x, r_y, \phi]^T \in \R^3$ and control input $u= [v, \omega]^T\in \R^2$
as close to the $r_x=0$ axis as possible without ever exceeding it, by minimizing the cost $|{r_x}|^2 + 10^{-6}\|u\|^2$ while satisfying the (soft) constraint $r_x \ge 0$ at all times. We discretize the system with RK4 ($dt=0.3$) and add the discrete disturbance, distributed according to $\mathcal{N}(0_3, \Sigma_\eta)$ for the unimodal and $0.8\,\mathcal{N}([-0.1, 0, 0]^T, \Sigma_\eta) + 0.2\,\mathcal{N}([0.4, 0, 0]^T, \Sigma_\eta)$ for the bimodal experiment, with $\Sigma_\eta=10^{-3}I_3$. Once again, we use $u_t = \pi(x_t, \theta_t)= \theta_t$ and the constraint $|\theta_t|\le \mathbf{1}_{n_u}$, set the horizon to $N=15$ and $1-\epsilon_k = 0.99$.

\addtolength{\textheight}{-0.5cm}
The results are summarized in Figure \ref{fig_pr_res} and Table \ref{tab:mpc_results_pr}, where the equivalent quantities as in the first example are reported. Notice in Figure \ref{fig_pr_res} how we are able to approach much closer to the $r_x=0$ line while keeping violations to the minimum, in contrast to the Gaussian MPC that stops relatively far from the desired target, in order to prioritize safety at the cost of conservatism. The Nominal MPC once again fails to provide a valid trajectory for both disturbance distributions, while the Scenario approach also is unable to yield a safe trajectory for the bimodal noise case. Our approach appears to be more efficient than the Scenario approach in the unimodal case, where the latter seems to require more than $20$ scenarios to reach the same level of constraint satisfaction whereas our approach needed just $m=3$ locations. \revisedText{blue}{Note that, although the constraints here are affine, we observe a larger closed-loop constraint violation than expected. We attribute this to the open-loop nature of the guarantees of Theorem \ref{thm_conv}, which moreover requires a sufficiently large $M$, and thus an $m$ potentially larger than the practical choice of $m=3$.}\looseness=-1

\revisedText{blue}{In theory, both the Scenario approach and ours will converge to the optimal solution as $N_s$ and $M$ grow,
at increasing computational cost.} However, notice the QP timings in Tables \ref{tab:mpc_results_dv} and \ref{tab:mpc_results_pr}, which indicate that our proposed framework is solved faster as it results in smaller optimization problems, where one constraint is imposed for all mixture components, in contrast to the per-scenario constraints needed for the Scenario-based approach. Finally, the overhead in total time is mainly due to the linearization and forward propagation of the mixtures; \revisedText{blue}{the per-component evaluations of $f_\pi,~\ell_\pi,~ h_\pi$, of their sensitivities, and the $|S_{k|t}|$ mode searches, are independent at each stage and thus directly parallelizable, which we expect to substantially reduce this overhead.}


%% file: sections/conclusion.tex
This work presented a tractable SNMPC reformulation for multi-modal additive disturbances using Gaussian mixture distributions for the predicted state. Utilizing Wasserstein bounds on the approximation error, we provided formal guarantees of correctness for our approach and quantified the sub-optimality of the reformulated problem's optimal solution in the context of the original SNMPC, for the case of quadratic-affine cost-constraint pairs. Numerical results on both tasks confirm that the proposed approach \revisedText{blue}{compares favorably with common SNMPC methods on the trade-off between optimality, scalability, and minimizing chance constraint violations, both for unimodal and multi-modal disturbances}. Future work will focus \revisedText{blue}{on establishing closed-loop guarantees, on improving the scalability and robustness of the solver implementation}, but also on extending the method to state-dependent noise and settings with distributional uncertainty.\looseness=-1